\documentclass[12pt]{article}

\usepackage{amstext,amsmath,amssymb,amsfonts,bbm}
\usepackage[latin1]{inputenc}
\usepackage{hyperref}
\usepackage{amsthm}
\usepackage{color}
\usepackage{graphicx}
\usepackage{bbold}

\usepackage{mathtools}
\DeclarePairedDelimiter\ceil{\lceil}{\rceil}

\theoremstyle{plain}

\newtheorem{lemma}{Lemma}
\newtheorem{theorem}{Theorem}
\newtheorem{proposition}{Proposition}
\theoremstyle{definition}

\DeclareMathOperator{\Tr}{Tr}

\begin{document}

\newcommand{\al}{\alpha}
\newcommand{\lmbd}{\lambda}
\newcommand{\tht}{\theta}
\newcommand{\Norm}{\Big\|}
\newcommand{\myI}{\mathrm{i}}

\newcommand{\cA}{{\cal A}}
\newcommand{\cZ}{{\cal Z}}
\newcommand{\cC}{{\cal C}}

\title{Variational Loop Vertex Expansion}

\author{Vasily Sazonov\footnote{vasily.sazonov@cea.fr}, \\Universit\'e  Paris-Saclay,  CEA,  List,\\  \textit{F-91120,  Palaiseau,  France}}

\maketitle

\begin{abstract}
Loop Vertex Expansion (LVE) was developed to construct QFT models with local and non-local interactions. Using LVE, one can prove the analyticity in the finite cardioid-like domain in the complex plain of the coupling constant of the free energies and cumulants of various vector, matrix, or tensor-type models. Here, applying the idea of choosing the initial approximation depending on the coupling constant, we construct the analytic continuation of the free energy of the quartic matrix model beyond the standard LVE cardioid over the branch cut and for arbitrary large couplings.
\end{abstract}

\section{Introduction}
The present work stems from the constructive field theory method of the Loop Vertex Expansion (LVE) and ideas of the variational perturbation theory.
The concept of LVE was first introduced in \cite{Rivasseau:2007fr} as a constructive approach for quartic matrix models aimed to provide bounds that are uniform in the matrix size. In its original form, LVE combines an intermediate field representation with replica fields and a forest formula \cite{BK, AR1} to express the free energy of the theory through a convergent sum over trees. Unlike conventional constructive methods, this loop vertex expansion does not rely on cluster expansions and does not entail conditions related to small/large field considerations.

Like Feynman's perturbative expansion, the LVE provides a straightforward way to calculate connected quantities. In this method, the theory's partition function is represented as a sum over forests, and its logarithm is essentially the same sum but constrained to connected forests -- trees. This property arises from the fact that the amplitudes factorize over the connected components of the forest.
The functional integrands associated with each forest or tree exhibit absolute and \emph{uniform} convergence for all field values.
Together with the non-proliferation of trees (in comparison to Feynman diagrams) this leads to the convergence of the LVE in the 'pacman'-like or cardioid-like domains, see Fig. \ref{F1}.
\begin{figure}[!ht]
\begin{center}
{\includegraphics[width=14cm]{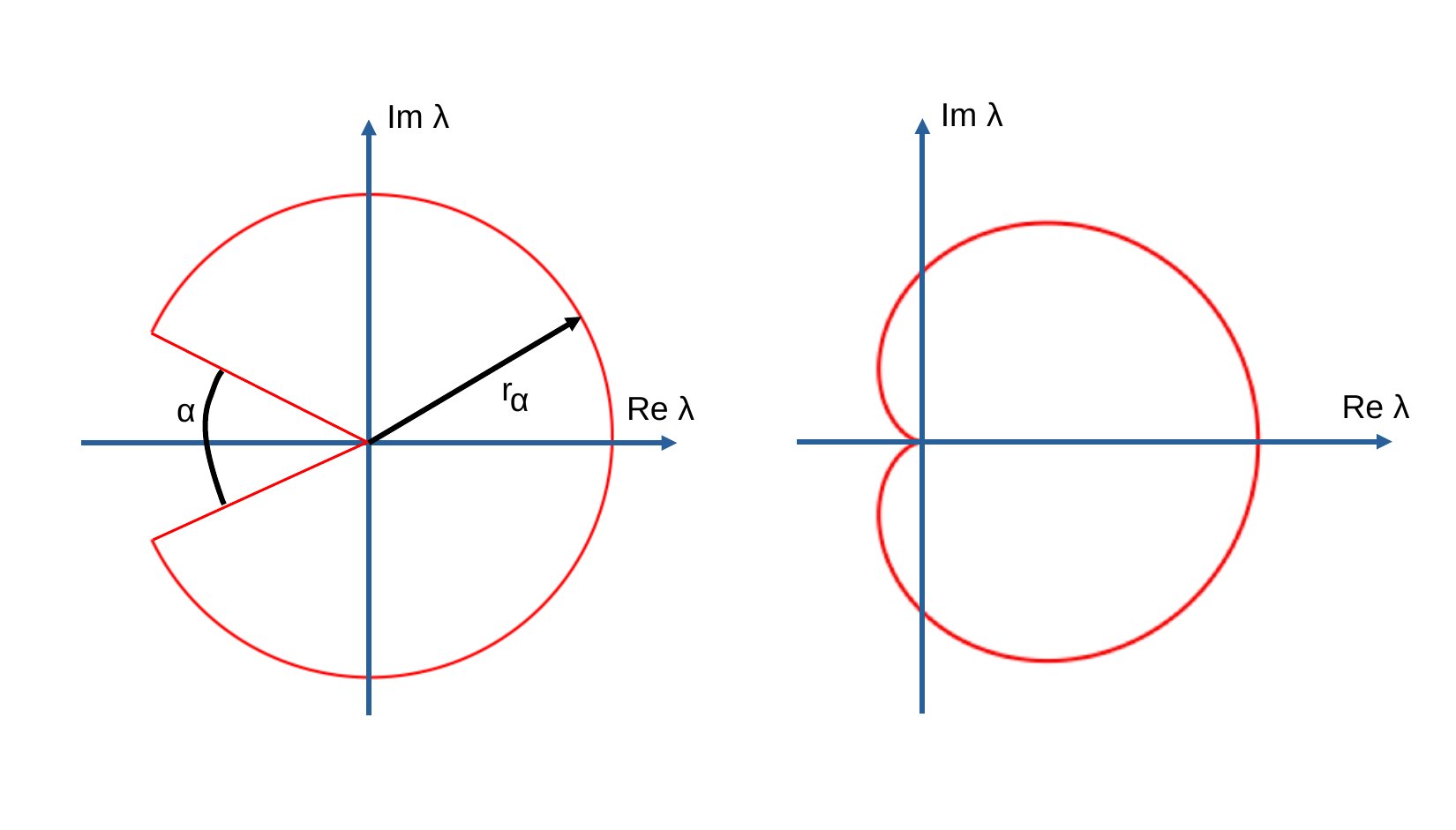}}
\end{center}
\caption{The left-hand side represents the typical LVE 'pacman'-type domain of analyticity with the radius $r_\al$ shrinking with decreasing of the angle $\al$. On the right, we present the cardioid domain, which can be understood as a union of 'pacman'-type domains.}
\label{F1}
\end{figure}
The convergence of the LVE implies the Borel summability of the standard perturbation series \cite{Sok}, and the LVE directly computes the Borel sum.
Essentially, the loop vertex expansion performs an \emph{explicit repacking} of infinitely many subsets of Feynman amplitude components, leading to a convergent expansion rather than a divergent one \cite{Rivasseau:2013ova}.

In the context of combinatorial field theories involving matrices and tensors \cite{Gurau:2011xp,Guraubook}, appropriately rescaled to possess a non-trivial $N \to \infty$ limit \cite{Hooft,Gurau:2010ba,Gurau:2011aq,Gurau:2011xq}, the Borel summability achieved by the LVE is \emph{uniform} with respect to the model's size $N$,
\cite{Rivasseau:2007fr,Gurau:2014lua,Gurau:2013pca,Delepouve:2014bma}.
The LVE method is extendable to ordinary field theories with cutoffs, as discussed in \cite{MR1}. An adapted multiscale version, known as MLVE \cite{Gurau:2013oqa}, incorporates renormalization techniques \cite{Delepouve:2014hfa,Lahoche:2015zya,Rivasseau:2017xbk,Rivasseau:2016rgt,Rivasseau:2014bya}, although it's worth noting that models developed so far are limited to the superrenormalizable type. The MLVE is particularly effective in resumming renormalized series for non-local field theories of matrix or tensorial types. 
Originally developed for the quartic interactions, the LVE method was then generalized to higher-order interactions under the name of the Loop Vertex Representation \cite{Rivasseau:2017hpg, Krajewski1, Krajewski2}.

The main principle of the Variational Perturbation Theory (VPT) is the utilization of the initial approximation chosen depending on the coupling constant and/or on the order of the perturbative expansion, or optimized depending on these parameters. Such ideas of the VPT were applied under different names: Variational Perturbation Theory \cite{Feynman, Kleinert1}, Optimized Perturbation Theory \cite{Stevenson}, Convergent Perturbation Theory \cite{Ushveridze}, Delta Expansion \cite{Guida}, etc.,  -- with the different levels of rigor to a various range of physical problems. The applications of VPT include:
energy levels of quantum anharmonic oscillator \cite{Caswell, Halliday, Guida};
critical indices in scalar quantum field theories \cite{KleinertCrit};
optimization of the QCD perturbative computations \cite{Stevenson, Stevenson2};
computations in lattice models with real and complex scalar actions \cite{Ivanov1, Ivanov2, ComplexAction};
coefficients of the $\frac{1}{D}$ expansion in the vector model \cite{largeD}. VPT was used for constructing strong coupling expansions \cite{strong}. It was shown that VPT (Delta Expansion) should be applicable in cases where the standard perturbation theory is non-Borel summable \cite{Guida}. It is worth noting that the proof of the convergence of the VPT (Delta Expansion) for the double-well case of the anharmonic oscillator is based on the analyticity of its energy levels \cite{Simon} in the area of the coupling constant Riemann surface, which is larger than the one required for the Borel summability in the single-well potential case.

In the current work, we combine the LVE method with the VPT idea of the initial approximation depending on the coupling constant. Taking the quartic matrix model as an example, we construct LVE with the modified initial approximation depending on the coupling constant and prove the convergence of the corresponding series for the free energy of the model for arbitrary large 
coupling constants with the arguments $-\frac{3\pi}{2} < \phi < \frac{3\pi}{2}$. 
The latter result is not optimal and might be improved up to  $-2\pi + \epsilon < \phi < 2\pi - \epsilon$, where $\epsilon > 0$.
This and the extension of the current results to cumulants and a constructive version of the $1/N$ expansion are additional outcomes of the method, left for further exploration. We also aim to investigate connections with the resurgence theory in the spirit of the recently carried out analysis of the vector model \cite{Benedetti}.

\section{Statement of the result}
To illustrate the main advantages of the Loop Vertex Expansion modified by the initial approximation depending on the coupling, we study a quartic matrix model. The partition function of this model 
is defined by
\begin{eqnarray}
\cZ[\lmbd,N]=
\int dM\exp\Big\{
-\Tr(MM^{\dagger})-\frac{\lmbd}{2N}\Tr([MM^{\dagger}]^2)
\Big\}\,,
\label{eq1}
\end{eqnarray}
where $M$ are complex $N\times N$ matrices.
The measure $dM$ is given by
\begin{equation}
dM=\pi^{-N^2}\prod_{1\leq i,j\leq N}d\text{Re}(M_{ij})d\text{Im}(M_{ij}) \, .
\end{equation}
The core object of our studies here is the free energy of the model defined, as
\begin{eqnarray}
    F[\lmbd, N] = -\frac{1}{N^2} \log \cZ[\lmbd, N]\,.
\label{eq2}
\end{eqnarray}
As was mentioned in the introduction the standard LVE tools allow one to prove the analyticity of the free energy \eqref{eq2} or of the cumulants of the model in the 'pacman'-like or cardioid-like domains, as in Fig. \ref{F1}. The sharpest result regarding the model \eqref{eq1} was obtained by LVE in \cite{Gurau:2014lua}, among other results it was shown there that the free energy \eqref{eq2} is analytic in the
cardioid domain of the coupling constant $\lmbd$,
\begin{equation}
  \cC_0 = \Big\{\lmbd\in \mathbb{C} \,\Big| \arg\lmbd = \phi\,,  4|\lmbd|< \cos^{2}\Big({\frac{\phi}{2}}\Big)\Big\} \, .
\label{cardioid}
\end{equation}
Obviously, instead of considering $\lambda \in \mathbb{C}$ in \eqref{cardioid}, we may define the cardioid as a subset of the Riemann surface of $\sqrt{\lambda}$, ${\cal Y}$, namely,
\begin{equation}
  \cC = \Big\{\lmbd\in {\cal Y} \,\Big| \arg\lmbd = \phi\,,  4|\lmbd|< \cos^{2}\Big({\frac{\phi}{2}}\Big), |\phi| < \pi \Big\} \, .
\label{cardioid2}
\end{equation}

Then, the following theorem gives the main result of the current paper, obtained by merging the ideas of the variational perturbation theory and loop vertex expansion.
\begin{theorem}
\label{mainTheorem}
    For any $\lmbd \in {\cal X}$, where ${\cal X} = \cC \cup {\cal E}$, and ${\cal E}$ is defined as a subset of Riemann surface of $\sqrt{\lambda}$, ${\cal Y}$, with $\lmbd \neq 0$ and $|\arg\lmbd| < \frac{3\pi}{2}$,
    the free energy $F[\lmbd, N]$ of the model \eqref{eq1},
    is analytic in $\lmbd$ uniformly in $N$.
\end{theorem}

\section{New initial approximation and intermediate field representation}

Hereafter we assume that $\lmbd \neq 0$. The case of $\lmbd = 0$ is trivial and the vicinity of $\lmbd = 0$ can be easily treated with the standard LVE tools \cite{Gurau:2014lua}. As a first step of our program, we change the Gaussian part of the action which we treat as an unperturbed measure, in other words, we shift the initial approximation from $\Tr(MM^\dagger)$ to $a\Tr(MM^\dagger)$. Then, introducing the intermediate field representation, we obtain,
\begin{eqnarray}
&&\cZ[\lmbd,N]=
\int dM \int dA\, \exp \Big\{-\frac{1}{2}\Tr(A^{2})
- a \Tr(MM^{\dagger})\nonumber\\
&&+ \myI \Tr\Big(A \big[\sqrt{\frac{\lmbd}{N}}MM^{\dagger} + \frac{(1 - a) \sqrt{N}}{\sqrt{\lmbd}}\mathbb{1}\big]\Big) + \Tr\big[\frac{(1 - a)^2N}{2 \lmbd}\mathbb{1}\big]
\Big\}\,,
\label{eq3}
\end{eqnarray}
where $Re\,a >0$, $A$ field is realised by the $N\times N$ Hermitian matrices, and the integral over $A$ is assumed to be normalized, so
\begin{equation}
    \int dA \exp\Big\{-\frac{1}{2} \Tr(A^2)\Big\} = 1\,.
\end{equation}
The integral over the initial degrees of freedom, matrices $M$ and $M^{\dagger}$, is Gaussian, therefore it can be evaluated. The corresponding covariance is given by
$\big(a-\myI\sqrt{\frac{\lmbd}{N}}A\big)\otimes \mathbb{1}$, and using that
\begin{eqnarray}
\det \left[ \Big(a-\myI\sqrt{\frac{\lmbd}{N}}A \Big)\otimes \mathbb{1} \right] 
=
a^{N^2} \exp \bigg\{N\Tr\log\Big(\mathbb{1}-\myI\sqrt{\frac{\lmbd}{a^2 N}}A\Big)\bigg\} \,,
\end{eqnarray}
we obtain
\begin{eqnarray}
\cZ[\lmbd,N]&=& 
e^{\frac{N^2 (1 - a)^2 }{2\lmbd}}a^{-N^2}
\int dA\,
\exp\bigg\{ - \frac{1}{2}\Tr(A^{2})\nonumber\\
&-&N\Tr\log\left(\mathbb{1}-\myI\sqrt{\frac{\lmbd}{a^2 N}}A\right)+
\myI\Tr\bigg(A \frac{(1 - a) \sqrt{N}}{\sqrt{\lmbd}}\bigg) \bigg\}\,.
\end{eqnarray}
Then, we define the non-polynomial interaction, as
\begin{eqnarray}
    {\cal S}[\lmbd, N, a](A) = \Tr\log\left(1-\myI\sqrt{\frac{\lmbd}{a^2 N}}A\right)-
\frac{\myI}{\sqrt{N}}\Tr\bigg(A \frac{(1 - a)}{\sqrt{\lmbd}}\bigg)\,,
\end{eqnarray}
and the normalization, as
\begin{eqnarray}
    K[\lmbd, N, a] = e^{\frac{N^2 (1 - a)^2 }{2 \lmbd}}a^{-N^2}\,.
\end{eqnarray}
Thus, the partition function can be written as
\begin{eqnarray}
\cZ[\lmbd,N]= 
K[\lmbd, N, a]
\int dA\,
\exp\bigg\{ - \frac{1}{2}\Tr(A^{2})
-N{\cal S}[\lmbd, N, a](A) \bigg\}\,.
\label{eq4}
\end{eqnarray}
Let us study the analytic properties of the partition function.
Hereafter, we write the parameter $a$ describing the initial approximation as,
\begin{eqnarray}
    a = x \sqrt{\lmbd} e^{\myI\psi}\,, \qquad -\frac{\pi}{2} < \psi + \frac{\phi}{2} < \frac{\pi}{2}\,, \qquad x > 0\,.
\label{arep}
\end{eqnarray}
Using these notations, we present the bound (14) from \cite{Gurau:2014lua}.
 \begin{lemma}
 \label{boundresolvent}
According to the notations above, we have $\frac{\lmbd}{a^2}=x\mathrm{e}^{-2\myI\psi}$ with $x>0$, then for $\psi \in (-\pi/2, \pi/2)$:
 \begin{equation}
 \Norm\Big(1-\myI\frac{\sqrt{\lmbd}}{a\sqrt{N}}A\Big)^{-1}\Norm\leq\frac{1}{\cos \psi} \, ,
 \end{equation}
 where $ \Norm \cdot \Norm $ stands for the operator norm.
 \end{lemma}
 
 \begin{proof}
 The first step is to rewrite the resolvent, as an integral,
\begin{equation}
\Big(1-\myI\frac{\sqrt{\lmbd}}{a\sqrt{N}}A\Big)^{-1}
=\frac{a\sqrt{N}}{\sqrt{\lmbd}}
\int_{0}^{\infty}d\al\,\exp\Big\{  -\al  \frac{a\sqrt{N}}{\sqrt{\lmbd}} + \al \myI A\Big\} \,.
\end{equation}
Using the latter representation, we immediately arrive at the bound for the operator norm,
\begin{eqnarray}
\Norm\Big(1-\myI\frac{\sqrt{\lmbd}}{a\sqrt{N}}A\Big)^{-1}
\Norm
&\leq&\frac{|a|\sqrt{N}}{|\sqrt{\lmbd}|}
\int_{0}^{\infty}\exp\Big\{-\al\text{Re}\big(\frac{a\sqrt{N}}{\sqrt{\lmbd}}\big)\Big\}
\Norm\exp\Big\{\myI \al A\Big\}\Norm\nonumber\\
&=&\frac{1}{\cos\psi} \, .
\end{eqnarray}
\end{proof}

Writing \eqref{eq4}, as
\begin{equation}
\cZ[\lmbd, N] = 
K[\lmbd, N, a]\int dA\,
\frac{\exp\bigg\{-\frac{1}{2}\Tr(A^{2})
-\myI \Tr\bigg(A \frac{(1 - a)\sqrt{N}}{\sqrt{\lmbd}}\bigg) \bigg\}}{\Big[\det \Big(1-\myI\frac{\sqrt{\lmbd}}{a\sqrt{N}}A
\Big)\Big]^{N}}\,,
\end{equation}
we are in a position to formulate the following proposition.
\begin{proposition}
$\cZ[N,\lmbd]$ is analytic in $\lmbd$ on the Riemann surface of the square root with $-2\pi < \phi < 2\pi$. 
\begin{proof}
Using lemma \ref{boundresolvent}, one can show that this integral is convergent for $\psi\in (-\pi/2,\pi/2)$ and $\lambda \neq 0$, or equivalently for $\lmbd / a^2\in{\mathbb C}-{\mathbb R}^{-}-\{0\}$. 
Since the integrand is analytic for $\psi\in (-\pi/2,\pi/2)$ and $\lambda \neq 0$, and $|\psi + \frac{\phi}{2}| < \frac{\pi}{2}$, taking into account proposition $5$ from \cite{Gurau:2014lua} for the case of $\lambda = 0$, we complete the proof.
\end{proof}
\end{proposition}

Note, that since $\cZ[N,\lmbd]$ may be zero for certain values of $\lmbd$, its analyticity does not imply the analyticity of the free energy.

\section{The Loop Vertex Expansion}
Using the representation \eqref{eq4} for the partition function, we expand the interaction part of the exponent into the Taylor series. Then, the summation and integration can be interchanged for $\lmbd / a^2\in{\mathbb C}-{\mathbb R}^{-}-\{0\}$, since the resulting series
\begin{eqnarray}
\cZ&=& 
K[\lmbd, N, a]\sum_{n=0}^{\infty}\frac{(-1)^{n}}{n!}
\int d\mu(A)\,
\bigg[N{\cal S}[\lmbd, N, a](A)\bigg]^n\, ,\\
d\mu(A)&=&dA\exp\{-\frac{1}{2}\Tr(A^{2})\}
\end{eqnarray}
is bounded by its absolute value.
Hereafter we drop the arguments $\lmbd$ and $N$ of $\cZ$, and write ${\cal S}[\lmbd, N, a](A)$ as ${\cal S}(A)$ to simplify the notations. 

Applying the replica trick, we replace (at each order $n$) the integral over a single matrix $A$ by the integral over a $n$-copies 
of $N\times N$ Hermitian matrices $A=(A_{i})_{1\leq i\leq n}$. The measure of the Gaussian integral over the replicated matrices $A$ is normalized with a degenerated covariance  $C_{ij}=1$. For any real positive symmetric matrix 
$C_{ij}$ the Gaussian integral obeys 
\begin{equation}
\int d\mu_{C}(A) \,A_{i|ab}A_{j|cd}=C_{ij}\,\delta_{ad}\delta_{bc} \,,~~~\int d\mu_{C}(A)  = 1\,,
\end{equation}
where $A_{i|ab}$ corresponds to the matrix element of $A_{i}$ in the row $a$ and column $b$.
The degenerated covariance in the Gaussian integral is equivalent to insertion $n-1$ 
Dirac $\delta$-functions $\delta(A_{1}-A_{2})\cdots\delta(A_{n-1}-A_{n})$. 
In Feynman diagrams, the uniform covariance connects the various replicas together (with the appropriate weights). 

Then, the partition function is given by
\begin{eqnarray}
\cZ=
K[\lmbd, N, a]\sum_{n=0}^{\infty}\frac{(-1)^{n}}{n!}
\int d\mu_{C}(A)
\prod_{i=1}^{n}
\bigg[N{\cal S}(A_i) \bigg]
\, .
\end{eqnarray}
To take the logarithm of the partition function, we are going to convert the latter expansion into the sum over forests by applying the Bridges-Kennedy-Abdessalam-Rivasseau formula \cite{BK, AR1}. The first step to do it is to replace
the covariance $C_{ij}=1$ by $C_{ij}(x)=x_{ij}$,  $x_{ij}=x_{ji}$ (which at the end should be evaluated at $x_{ij}=1$) for $i\neq j$ and $C_{ii}(x)=1$. Then,
\begin{eqnarray}
\cZ&=&
K[\lmbd, N, a]\sum_{F\,\text{labeled forest}}\frac{(-1)^{n}}{n!}\int_{0}^{1} \prod_{(i,j)\in F}dt_{ij}  \,\,   \left( \prod_{(i,j)\in F}\frac{\partial }{\partial x_{ij}} \right) \nonumber\\
&\times& \bigg\{
\int d\mu_{C(x)}(A)\prod_{i=1}^{n}
\bigg[N{\cal S}(A_i) \bigg]\bigg\}\bigg|_{x_{ij}=v^F_{ij}} \, ,\nonumber\\
\label{LVEcombinatorial}
\end{eqnarray}
 where $n$ is the number of vertices of the forest $F$, $i$ and $j$ label the forest's vertices, and there is a weakening parameter
 $t_{ij}$ per each edge $(i,j)$ of the forest, and 
 \begin{equation}
v^{F}_{ij}=\left\{\begin{array}{ccl}
\inf_{(k,l)\in{ P}_{i\leftrightarrow j}^{{F}}} t_{kl}&\text{if}&  { P}_{i\leftrightarrow j}^{{F}} \, \text{exists} \\
0&\text{if}&  { P}_{i\leftrightarrow j}^{{F}} \, \text{does not exist} \
\end{array}\right. \, ,
\end{equation}
 ${ P}_{i\leftrightarrow j}^{{F}} $ is the unique path in the forest $F$ joining $i$ and $j$ (the infimum is taken to be $1$ if $i=j$). 
 
Applying the following lemma, \cite{Gurau:2014lua}, we can take the logarithm.
\begin{lemma}
Let ${\cal W}(T)$ be the weight of a tree $T$, not depending on the labels of the tree vertices and the weight of a forest 
${\cal W} (F)$ is defined to be the product of the weights of its trees. Then, in the formal series sense, we have
\begin{equation}
\log \sum_{F \text{ labeled forests}}\frac{{\cal W}(F)}{|V(F)|!}=
\sum_{T\text{ labeled trees}}\frac{{\cal W}(T)}{|V(T)|!} \, ,
\end{equation}
where $|V(F)|$ and $|V(T)|$ are the number of vertices in $F$ and $T$ correspondingly.
\end{lemma}
Since the differentiation with respect to $x_{ij}$ and the Gaussian integration factor over the trees in the forest $F$, we obtain
 \begin{eqnarray}
  \log\cZ &=& \log K[\lmbd, N, a] + \sum_{T\,\text{labeled trees}}\frac{(-1)^{n}}{n!}\int_{0}^{1} \prod_{(i,j)\in T}dt_{ij} \, 
\left( \prod_{(i,j)\in T}\frac{\partial }{\partial x_{ij}} \right) 
\crcr
&\times& \bigg\{
\int d\mu_{C(x)}(A)\prod_{i=1}^{n}\bigg[N{\cal S}(A_i) \bigg]\bigg\}\bigg|_{v^T_{ij}} \, ,\crcr
 v^{T}_{ij} &=& \inf_{(k,l)\in{ P}_{i\leftrightarrow j}^{T} } t_{kl} \, .
\end{eqnarray}
where $n = |V(T)|$ is the number of the vertices, and ${ P}_{i\leftrightarrow j}^{T} $ stands for the unique path joining vertices $i$ and $j$ in the tree $T$. 
Expressing the Gaussian integral as a differential operator,
\begin{eqnarray}
 \int d\mu_{C(x)}(A)  F(A) =
 \left[ e^{\frac{1}{2} \sum_{i,j} x_{ij} \Tr \left[\frac{\partial }{ \partial A_i} \frac{\partial }{ \partial A_j}\right] } F(A)\right]_{A_i=0} \, ,
\end{eqnarray}
 we see that
\begin{equation}
\frac{\partial}{\partial x_{ij}}
\bigg(\int d\mu_{C(x)}(A)F(A)\bigg)=
\frac{1}{2} \int d\mu_{C(x)}(A) \, \Tr \left[\frac{\partial }{ \partial A_i} \frac{\partial }{ \partial A_j}\right] F(A) \, .
\end{equation}
The latter differential operator acts on $i$ and $j$ vertices and connects them by an edge. 

The first derivative of the loop vertex (non-polynomial action ${\cal S}(A_i)$) is given by
\begin{eqnarray}
   &&\frac{\partial}{\partial A_{i|cd}} \bigg[\Tr\log\left(1-\myI\sqrt{\frac{\lmbd}{a^2 N}}A\right)+
\frac{\myI}{\sqrt{N}}\Tr\bigg(A \frac{(1 - a)}{\sqrt{\lmbd}}\bigg) \bigg] = \nonumber\\
&&-\myI\frac{\sqrt{\lmbd}}{a\sqrt{N}} \bigg(1-\text{i}\frac{\sqrt{\lmbd}}{a\sqrt{N}}A_{i}\Big)_{dc}^{-1} + \myI \frac{(1 - a)}{\sqrt{\lmbd} \sqrt{N}} \mathbb{1}_{dc}\,.
\label{firstD}
\end{eqnarray}
Only the first term of \eqref{firstD} is relevant for applying all further derivatives. Therefore all other derivatives can be computed using the following recursive relation,
\begin{equation}
\frac{\partial }{ \partial A_{i|ab}} \frac{-\myI\sqrt{\lmbd}}{a\sqrt{N}} \Big(1-\text{i}\frac{\sqrt{\lmbd}}{a\sqrt{N}}A_{i}\Big)_{cd}^{-1}
= -\frac{\lmbd}{a^2 N} \Big(1-\text{i}\frac{\sqrt{\lmbd}}{a\sqrt{N}}A_{i}\Big)_{ca}^{-1} \Big(1-\text{i}\frac{\sqrt{\lmbd}}{a\sqrt{N}}A_{i}\Big)_{bd}^{-1} \, .
\end{equation}
Let $V(T)$ be the set of all vertices of the tree, and $E(T)$ be the set of edges of the tree.
We observe that in each tree $T$ with $|V(T)| > 2$ there are two types of vertices: 
internal vertices and leaves. If the vertex is differentiated only once, it is a leaf, bringing the contribution of the form of \eqref{firstD}. We say that leaf vertices have only one corner. Multiple derivatives acting on the same vertex (corresponding to multiple edges hooked to it) can act on either of the corners of the vertex, splitting it into two corners and eliminating the constant term of \eqref{firstD}, see Fig. \ref{LVEcorners}. 
\begin{figure}[!ht]
\begin{center}
{\includegraphics[width=7cm]{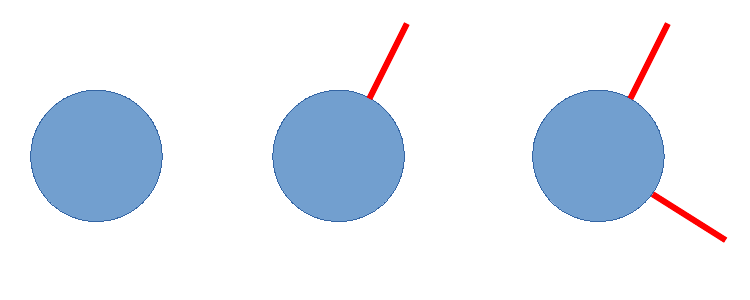}}
\end{center}
\caption{From left to right: a vertex without corners -- corresponds to a trivial tree, a vertex with only one corner -- a leaf, vertex with two corners.}
\label{LVEcorners}
\end{figure}
Each derivative brings an additional factor of $\frac{1}{\sqrt{N}}$, and consequently, each edge of any tree comes with the factor $\frac{1}{N}$. In the following we attribute factors $\frac{1}{N}$ to the edges, and define a corner operator as
\begin{equation}
\cC = 
    \begin{cases}
      \frac{-\myI\sqrt{\lmbd}}{a}\Big(1-\text{i}\sqrt{\frac{\lmbd}{a^2 N}}A_{i}\Big)_{dc}^{-1} + \myI \frac{(1 - a)}{\sqrt{\lmbd}} \mathbb{1}_{dc}\,, \text{   only one corner in the vertex}\\~\\
      \frac{-\myI\sqrt{\lmbd}}{a}\Big(1-\text{i}\sqrt{\frac{\lmbd}{a^2 N}}A_{i}\Big)_{dc}^{-1}\,, \text{   if there are more corners}
    \end{cases}\,.
\end{equation}
Therefore, the logarithm of $\cZ$ can be expressed, as
\begin{align}
 \log \cZ  = & \log K[\lmbd, N, a] +
\sum_{T\,\atop\text{LVE tree}}  \cA_{T}[\lmbd,N]  \, ,\crcr
\cA_{T}[\lmbd,N]   = &  (-1)^{|V(T)|}\frac{N^{|V(T)|-|E(T)|}}{|V(T)|!} \int_0^1 \prod_{e\in E(T)}dt_{e}\, \crcr
  & \times \int d\mu_{C_{T}}(A) \, 
 \Tr\Big[
\mathop{\overrightarrow{\prod}}\limits_{c\in \partial T\,\text{corner}}
\cC_c(i_c)\Big] \, ,
\label{eqLVE}
\end{align}
where $\partial T$ denotes the set of corners of the tree $T$, $i_{c}$ labels the vertex, to which the corner $c$ is attached,the covariance $C_T$ is 
\begin{equation}
(C_{T})_{ij}=\inf_{(k,l)\in P^{T}_{i\leftrightarrow j}} t_{kl}\,,
\end{equation}
and the infimum is taken to be $1$ if $i=j$.

\section{Bounds}
The expansion \eqref{eqLVE} contains three types of contributions, which should be considered separately: a trivial tree with only one vertex, a tree with two vertices, and all other trees. In the following we treat these cases step by step.

\subsection{Trivial tree}
As always in the LVE formalism there is a trivial tree with only one vertex, see Fig. \ref{LVEcorners}. It requires a special treatment. The contribution of the trivial tree is given by
\begin{eqnarray}
    \cA_{T_1} = \int d\mu_{C(x)}(A) \bigg[N \Tr\log\left(1-\myI\frac{\sqrt{\lmbd}}{a\sqrt{N}}A\right)+
\myI \Tr\bigg(A \frac{(1 - a)\sqrt{N}}{\sqrt{\lmbd}}\bigg) \bigg]\,.
\end{eqnarray}
The second term gives zero after the integration and the first can be bounded by integrating by parts as 
\begin{eqnarray}
    \cA_{T_1} &=& N \int dA\, e^{-\frac{1}{2}\Tr\big[A^2\big]} \Tr\Bigg[ \log (\mathbb{1} - \text{i} \frac{\sqrt{\lmbd}}{a\sqrt{N}}A)\Bigg]\nonumber\\
    &=& N \int dA\, e^{-\frac{1}{2}\Tr\big[A^2\big]} \Tr\Bigg[\int_0^1 dt\, \frac{- \text{i} \frac{\sqrt{\lmbd}}{a\sqrt{N}} t A}{(\mathbb{1} - \text{i} \frac{\sqrt{\lmbd}}{a\sqrt{N}}A)}\Bigg] \nonumber\\
    &=& \int d\mu(A)\, \Tr\Bigg[ \int_0^1 dt\, \frac{-\frac{\lmbd}{a^2} t}{(\mathbb{1} - \text{i} \frac{\sqrt{\lmbd}}{a\sqrt{N}}A)^2}\Bigg]
\end{eqnarray}
Consequently, we have the bound,
\begin{eqnarray}
    |\cA_{T_1}| \leq \frac{N}{2} \Big\vert\frac{\lambda}{a^2}\Big\vert \,.
\label{boundAT1}
\end{eqnarray}

Since the integrand is analytic for $\psi\in (-\pi/2,\pi/2)$, and $|\psi + \frac{\phi}{2}| < \frac{\pi}{2}$, employing \eqref{arep},  we arrive at the following conclusion.
\begin{lemma}
\label{trivialTreeLemma}
The amplitude of the trivial tree, $\cA_{T_1}$, is analytic in $\lmbd$ and uniformly in $N$ bounded by \eqref{boundAT1} on the Riemann surface of the square root with $\lambda \neq 0$ and $-2\pi < \phi < 2\pi$. 
\end{lemma}

\subsection{Two-vertex tree}
To derive the bound for the tree with two vertices, see Fig. \ref{T2fig}, we start with a general bound for the corner operators that are also useful for all other trees.
Using the lemma \ref{boundresolvent}, we obtain
\begin{equation}
\Vert\cC\Vert \leq 
    \begin{cases}
      \Bigg\vert \frac{\sqrt{\lmbd}}{a} \frac{1}{\cos\psi}\Bigg\vert + \Bigg\vert\frac{(1 - a)}{\sqrt{\lmbd}}\Bigg\vert\,, \text{   if there is only one corner in the vertex}\\~\\
      \Bigg\vert \frac{\sqrt{\lmbd}}{a} \frac{1}{\cos\psi}\Bigg\vert\,, \text{   if there are more corners}
    \end{cases}\,.
\label{cornerbound}
\end{equation}
\begin{figure}[!ht]
\begin{center}
{\includegraphics[width=3cm]{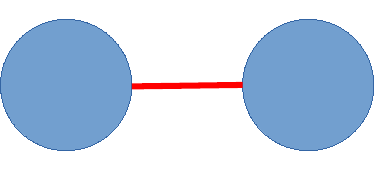}}
\end{center}
\caption{Two-vertex tree.}
\label{T2fig}
\end{figure}
Remembering the representation \eqref{arep}, for sufficiently large $x$, we can simplify the one-corner bound in \eqref{cornerbound}, requiring that
\begin{eqnarray}
    \Bigg\vert \frac{\sqrt{\lmbd}}{a} \frac{1}{\cos\psi}\Bigg\vert + \Bigg\vert\frac{(1 - a)}{\sqrt{\lmbd}}\Bigg\vert \leq  \Bigg\vert\frac{3(1 - a)}{2\sqrt{\lmbd}}\Bigg\vert\,.
\label{doubleBound}
\end{eqnarray}
It is not hard to see that the latter inequality is valid when
\begin{eqnarray}
    x \geq \max\{x_1, \frac{1}{|\sqrt{\lambda}|}\}\,,\qquad x_1 = \frac{\cos\psi + \sqrt{\cos^2\psi + 8|\lmbd|\cos\psi}}{2 |\sqrt{\lmbd}|\cos\psi}\,.
\label{xineq1}
\end{eqnarray}
Note, as $|\psi| < \pi/2$, we have that $\cos\psi>0$.

In the two-vertex tree, both vertices have only one corner, therefore, using \eqref{doubleBound}, taking into account that there are $2$ vertices, $1$ edge, and one factor $N$ coming from the trace, we can bound the amplitude of this tree as
\begin{eqnarray}
    \vert\cA_{T_2}\vert \leq \frac{9 N^2}{8} \int d\mu_{C(x)}(A) \Bigg\vert\frac{(1 - a)}{\sqrt{\lmbd}}\Bigg\vert^2\,.
\end{eqnarray}
Due to the analyticity of the integrand for $|\psi| < \pi/2$ and the latter bound, we arrive at the following.
\begin{lemma}
\label{T2Lemma}
The amplitude of the trivial tree, $\cA_{T_2}$, is analytic in $\lmbd$ and bounded on the Riemann surface of the square root with $\lmbd \neq 0$, $-2\pi < \phi < 2\pi$. 
\end{lemma}

\subsection{Other trees}

Obviously, the bound for the leaves is larger than for the internal vertices. If we use it for all vertices, we can bound the amplitude of each tree, as
\begin{align}
\bigg|
\Tr\Big[
\mathop{\overrightarrow{\prod}}\limits_{c\in \partial T\,\text{corner}}
\cC_c(i_c)\Big]
\bigg|
&\leq N  \Bigg(\Bigg\vert \frac{\sqrt{\lmbd}}{a} \frac{1}{\cos\psi}\Bigg\vert + \Bigg\vert\frac{(1 - a)}{\sqrt{\lmbd}}\Bigg\vert\Bigg)^{2|E(T)|}\, .
\label{eqTreeABound}
\end{align} 
However, the best results given by this bound are achieved when one takes $a = 1$, which corresponds to classic LVE results. It happens since it is impossible to make simultaneously small both terms: $\Bigg\vert \frac{\sqrt{\lmbd}}{a} \frac{1}{\cos\psi}\Bigg\vert \text{ and }\Bigg\vert\frac{(1 - a)}{\sqrt{\lmbd}}\Bigg\vert$.
To overcome this difficulty, we observe that:
\begin{itemize}
    \item We can always make the term $\Bigg\vert \frac{\sqrt{\lmbd}}{a} \frac{1}{\cos\psi}\Bigg\vert$ as small as needed, by varying the ratio $\frac{\sqrt{\lmbd}}{a}$ (considering large $x$).
    \item The factor $\Bigg\vert \frac{\sqrt{\lmbd}}{a} \frac{1}{\cos\psi}\Bigg\vert + \Bigg\vert\frac{(1 - a)}{\sqrt{\lmbd}}\Bigg\vert$, for which hereafter we will use the bound \eqref{doubleBound}, comes only from the leave vertices, and in a general tree there are not so many leaves, see Fig. \ref{treesFig}.
\end{itemize}
\begin{figure}[!ht]
\begin{center}
{\includegraphics[width=14cm]{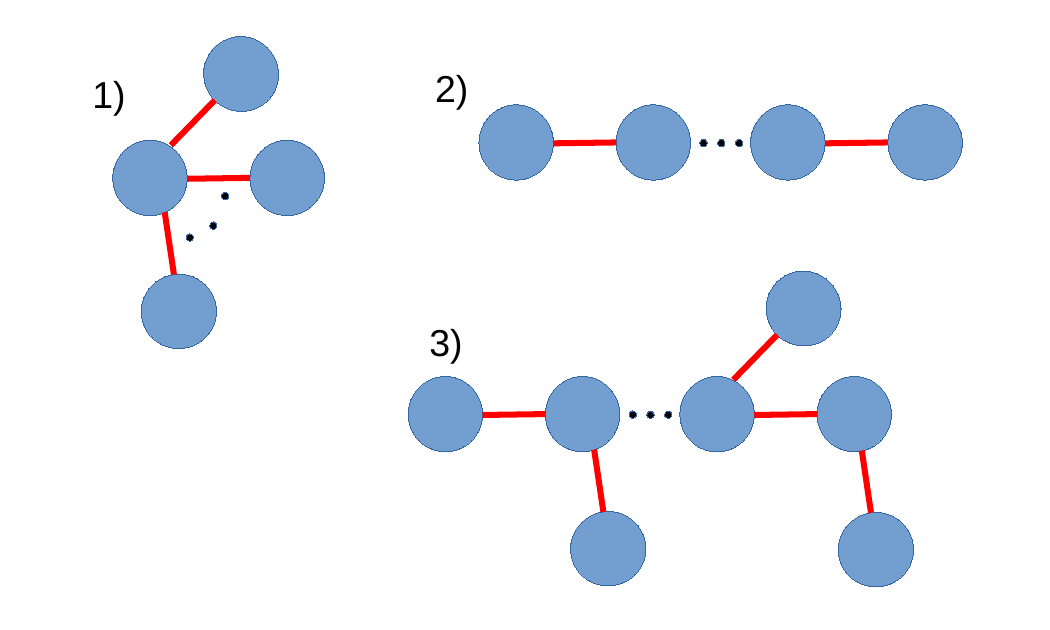}}
\end{center}
\caption{1) A tree with a maximal amount of leaves at the given order of the loop vertex expansion. 2) A tree with a minimal amount of leaves at the given order of the LVE. 3) Representation of an average LVE tree with not so many leaves.}
\label{treesFig}
\end{figure}

Being inspired by the latter observations, we first establish the bound for the corner operators of a general tree.
\begin{lemma}
\label{cornerLemma}
If a tree has $n_l$ leaves, $n_i$ internal vertices, and $n_l + n_i = |V(T)| > 2$, its amplitude given by a trace of the product of the corner operators is bounded by
\begin{align}
\bigg|
\Tr\Big[
\mathop{\overrightarrow{\prod}}\limits_{c\in \partial T_{(n_i, n_l)}\,\text{corner}}
\cC_c(i_c)\Big]
\bigg|
&\leq N  \Bigg\vert \frac{\sqrt{\lmbd}}{a} \frac{1}{\cos\psi}\Bigg\vert^{2 n_i - 2} \Bigg\vert\frac{3(1 - a)}{2 a \cos\psi}\Bigg\vert^{n_l}\, .
\label{eqTreeABound2}
\end{align} 
\end{lemma}
\begin{proof}
We prove \eqref{eqTreeABound2} by induction. For the base, it is enough to check \eqref{eqTreeABound2} for the tree with $n_i = 1$ and $n_l = 2$. According to \eqref{doubleBound}, two leaves bring the factor
\begin{eqnarray}
    \Bigg\vert\frac{3(1 - a)}{2\sqrt{\lmbd}}\Bigg\vert^2\,,
\end{eqnarray}
and two corner operators of the internal vertex give
\begin{eqnarray}
\Bigg\vert\frac{\sqrt{\lmbd}}{a} \frac{1}{\cos\psi}\Bigg\vert^{2}\,.
\label{twocorners}
\end{eqnarray}
All together, taking into account the factor $N$ coming from the trace, we have exactly \eqref{eqTreeABound2} for $n_i = 1$ and $n_l = 2$.

Then, for the induction, we assume that \eqref{eqTreeABound2} is valid for all trees with $n_v$ vertices, $n_l + n_i = n_v$, and prove that it is then valid for all trees with $(n_v + 1)$ vertices.
The trees with $(n_v + 1)$ vertices can be obtained from the trees with $n_v$ vertices by increasing the number of leaves or number of internal vertices. In the first case, see Fig. \ref{proof2}, the bound \eqref{eqTreeABound2} gets an extra factor
\begin{eqnarray}
    \Bigg\vert\frac{3(1 - a)}{2\sqrt{\lmbd}}\Bigg\vert\,,
\end{eqnarray}
from the leaf, and an additional factor
\begin{eqnarray}
\Bigg\vert\frac{\sqrt{\lmbd}}{a} \frac{1}{\cos\psi}\Bigg\vert
\end{eqnarray}
from the new corner of the internal vertex, where the new leaf is attached. Altogether, this gives \eqref{eqTreeABound2}
with the number of internal vertices $n_i$ and number of leaves $(n_l + 1)$.
\begin{figure}[!ht]
\begin{center}
{\includegraphics[width=14cm]{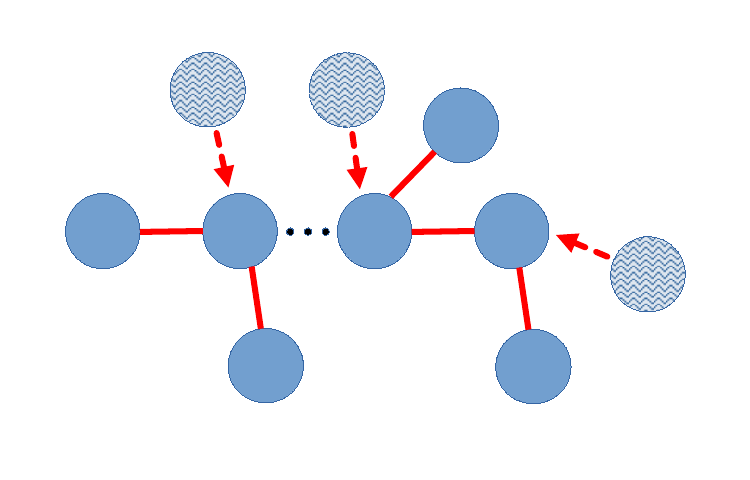}}
\end{center}
\caption{Possible ways to add a leaf to a tree.}
\label{proof2}
\end{figure}

In the second case, the number of vertices is augmented by increasing the number of the internal vertices.
Here, we first consider the case with $n_l = 2$, which corresponds to the chain of vertices.
Then, the insertion of the internal vertex 
brings the contribution from two corner operators of the new internal vertex, the same as \eqref{twocorners}. Together with the factor $N$ from the trace this gives \eqref{eqTreeABound2}
with the number of internal  $(n_i + 1)$ and number of leaves $n_l$.

When $n_l > 2$, there exists a tree $T(n_i + 1, n_l - 1)$ with $n_v$ vertices consisting of $(n_i + 1)$ internal vertices, and $(n_l - 1)$ leaves. According to the induction assumption, the bound \eqref{eqTreeABound2} is valid for this tree. Then, the tree of our current interest is obtained from $T(n_i + 1, n_l - 1)$ by inserting a leaf, and the insertion of a leaf was treated before.
\end{proof}

To prove the convergence of the series \eqref{eqLVE}, we split the sum over the LVE trees as
\begin{eqnarray}
    \sum_{T\,\atop\text{LVE tree}}  = \sum_{T \in \{T_1, T_2\}} + \sum_{2 < |V(T)| < 60} + \sum_{T_\geq} + \sum_{T_<}\,,
\label{sumsplit}
\end{eqnarray}
where the two terms of the first sum were treated before, the second sum runs over all LVE trees with $2 < |V(T)| < 60$ number of vertices, $T_\geq$ denotes the LVE trees which $\al |V(T)|$ leaves or more and $|V(T)| \geq 60$ vertices, and $T_<$ stands for the LVE trees with less than $\al |V(T)|$ leaves and $|V(T)| \geq 60$ vertices.

Let us now bound each sum of \eqref{sumsplit} separately. We start with the following lemma.
\begin{lemma}
\label{allbyConstLemma}
For any $x$, satisfying
\begin{eqnarray}
    x \geq x_2\,,\qquad x_2 = \frac{|\sqrt{\lmbd}| + 1}{|\sqrt{\lmbd}|}\,,
\label{x2cond}
\end{eqnarray}
one has
\begin{align}
\Bigg\vert \frac{\sqrt{\lmbd}}{a} \frac{1}{\cos\psi}\Bigg\vert \leq \Bigg\vert\frac{(1 - a)}{a \cos\psi} \Bigg\vert \leq \frac{\sqrt{2}}{\cos\psi} \,.
\label{twoIneq0}
\end{align}

\end{lemma}
\begin{proof}
Recall that $\sqrt{\lmbd} = e^{\text{i}\phi/2} |\sqrt\lmbd|$, $\lmbd \neq 0$, and that according to the representation \eqref{arep}, $a = x |\sqrt{\lmbd}| e^{\myI(\psi+\frac{\phi}{2})}$, $x > 0$, $-\frac{\pi}{2} < \psi  + \frac{\phi}{2} < \frac{\pi}{2}$.
Then, the first inequality in \eqref{twoIneq0} transforms to 
\begin{eqnarray}
    \frac{1}{x} \leq \Bigg\vert \frac{1 - |\sqrt{\lmbd}| x e^{\myI(\psi + \frac{\phi}{2})}}{|\sqrt{\lmbd}| x} \Bigg\vert\,,
\end{eqnarray}
and it is enough to satisfy it for $\psi + \frac{\phi}{2} = 0$. Simplifying the absolute value operation for $x \geq \frac{1}{|\sqrt{\lmbd}|}$, we arrive at \eqref{x2cond}. The second inequality in \eqref{twoIneq0} rewrites, as
\begin{align}
\Bigg\vert\frac{(1 - |\sqrt{\lmbd}| x e^{\myI(\psi + \frac{\phi}{2})})}{|\sqrt{\lmbd}| x \cos\psi} \Bigg\vert \leq \frac{\sqrt{2}}{\cos\psi} \,.
\label{twoIneq2}
\end{align} 
Using the upper bound for the left-hand side of \eqref{twoIneq2} (note that $-\frac{\pi}{2} < \psi  + \frac{\phi}{2} < \frac{\pi}{2}$), we find that it is enough to satisfy
\begin{align}
\Bigg\vert\frac{\sqrt{1 + |\lmbd| x^2}}{|\sqrt{\lmbd}|x} \Bigg\vert \leq \sqrt{2}\,,
\label{twoIneq3}
\end{align} 
what can be achieved for $x \geq \frac{1}{|\sqrt{\lmbd}|}$.
\end{proof}

\begin{lemma}
\label{FiniteSumLemma}
    For $|\psi| < \pi/2$, $\lambda \neq 0$ the sum of absolute values of trees amplitudes from the finite set of trees in \eqref{sumsplit} is bounded by a constant,
    \begin{eqnarray}
        \sum_{2 < |V(T)| < 60} |\cA_T| < const\,,
    \end{eqnarray}
where the $const \sim O(N^2)$.
\end{lemma}
\begin{proof}
    Follows from lemmas \ref{cornerLemma} and \ref{allbyConstLemma}.
\end{proof}

Now we can estimate how many trees have more than $\al |V(T)|$ leaves (we will be interested in $1/2 < \al < 1$).
\begin{lemma}
\label{LeavesLemma}
The number of trees with $|V(T)|$ vertices and $\al |V(T)|$ or more leaves with $1/2 < \al < 1$, $|T_\geq|$ is bounded by
\begin{eqnarray}
    |T_\geq| &\leq&  
    (|V(T)| - \ceil*{\al |V(T)|}  + 1)
    \big(|V(T)|!\big) \nonumber\\
    &\times& 2^{|V(T)|-3} e^{\ceil*{\al |V(T)|}} \bigg(\frac{1 - \al}{\al}\bigg)^{\ceil*{\al |V(T)|}}\,.
\label{alphbound}
\end{eqnarray}
\end{lemma}
\begin{proof}
The number of labeled trees with $n$ vertices and $k$ leaves is given by
\begin{eqnarray}
    {\cal N}(n, k) = \frac{n!}{k!} S(n-2, n-k)\,,
\end{eqnarray}
where $S(n-2, n-k)$ is the Stirling number of the second type.
We can bound ${\cal N}(n, k)$ using an upper bound for the Stirling's numbers of the second type,
\begin{eqnarray}
    S(n-2, n-k) \leq 
    \binom{n-3}{n-k-1} (n - k)^{k} \leq 2^{n-3} (n - k)^{k}\,,
\end{eqnarray}
and a lower bound for the Gamma function
\begin{eqnarray}
    k! \geq e^{-k} (k+1)^k\,.
\end{eqnarray}
Thus, we have
\begin{eqnarray}
    {\cal N}(n, k) \leq n!\frac{e^k}{k^k} 2^{n-3} (n-k)^k\,.
\end{eqnarray}
Applying it to $n = |V(T)|$, $k = \ceil*{\al |V(T)|}$ -- an integer number larger or equal to $\al |V(T)|$, we obtain
\begin{eqnarray}
    {\cal N}(|V(T)|, \ceil*{\al |V(T)|}) \leq \big(|V(T)|!\big) 2^{|V(T)|-3} e^{\ceil*{\al |V(T)|}} \bigg(\frac{1 - \al}{\al}\bigg)^{\ceil*{\al |V(T)|}}\,,
\label{NalphaB}
\end{eqnarray}
where we have used an obvious relation
\begin{eqnarray}
    \frac{|V(T)| - \ceil*{\al |V(T)|}}{\ceil*{\al |V(T)|}} \leq \frac{1 - \al}{\al}\,.
\end{eqnarray}
If we take $\frac{e}{1 + e} < \al < 1$, then $e\frac{1 - \al}{\al} < 1$,
and the expression decreases with increasing of $\al$ leading to decreasing of the bound \eqref{NalphaB}.
Therefore for large values of $\al$ (for larger amounts of leaves), we will have fewer trees. Consequently, we can bound the number of trees that have $\al |V(T)|$ leaves or more just by multiplying the \eqref{NalphaB} by $(|V(T)| - \ceil*{\al |V(T)|}  + 1)$.
\end{proof}

Now we are ready to bound the sum of amplitudes of all trees $T_\geq$ with $\al |V(T)|$ leaves or more. 
\begin{lemma}
\label{TgeqLemma}
For $|\psi| < \pi/4$ and $\al = \frac{59}{60}$, the sum of absolute values of amplitudes of the trees $T_\geq$ is smaller than the sum of an absolutely convergent series,
\begin{eqnarray}
    \sum_{T_\geq} \big\vert \cA_{T_\geq}\big \vert \leq \frac{N^2 e}{8}\sum_{v=60}^\infty (\frac{1}{60} v  + 1)\, \Bigg(\frac{7}{8}\Bigg)^{v}\,.
\end{eqnarray}
\end{lemma}
\begin{proof}
First of all, we chose $x \geq x_3$, $x_3 = \max\{x_1, x_2\}$.
Then, according to lemmas \ref{cornerLemma} and \ref{allbyConstLemma}, the contributions of the corner operators of each tree from $T_\geq$ can be bounded, as
\begin{align}
\bigg|
\Tr\Big[
\mathop{\overrightarrow{\prod}}\limits_{c\in \partial T_\geq\,\text{corner}}
\cC_c(i_c)\Big]
\bigg|
\leq N  \Bigg(\frac{3\sqrt{2}}{2\cos\psi}\Bigg)^{2 n_i + n_l - 2} \leq N  \Bigg(\frac{3\sqrt{2}}{2\cos\psi}\Bigg)^{2 |V(T)|}\, .
\label{eqTreeABoundTmanyLeaves}
\end{align} 
Taking into account the number of trees in $T_\geq$ at each order of LVE by lemma \ref{LeavesLemma}, we obtain
\begin{eqnarray}
    \sum_{T_\geq} \big\vert \cA_{T_\geq}\big \vert &\leq&
    \frac{N^2}{8}\sum_{v=60}^\infty (v - \ceil*{\al v}  + 1)\, 2^v e^{\ceil*{\al v}} \bigg|\frac{1 - \al}{\al}\bigg|^{\ceil*{\al v}} \Bigg(\frac{3\sqrt{2}}{2\cos\psi}\Bigg)^{2 v}\nonumber\\
    &\leq&
    \frac{N^2}{8}\sum_{v=60}^\infty (v - \al v  + 1)\, e^{\al v + 1} \bigg|\frac{1 - \al}{\al}\bigg|^{\al v} \Bigg(\frac{3}{\cos\psi}\Bigg)^{2 v}\,.
\end{eqnarray}
Now let us take $\al$ in such a way that
\begin{eqnarray}
    \Bigg(\frac{3}{\cos\psi}\Bigg)^2 \, e^{\al} \Bigg\vert\frac{1 - \al}{\al}\Bigg\vert^{\al} < 1\,.
\label{alphacos}
\end{eqnarray}
It is easy to see that, for any value of the $\cos\psi$, one can choose such $\al$'s, sufficiently close to $\al = 1$, that the inequality \eqref{alphacos} is satisfied. However, to handle our bounds we need to fix a certain constant value of $\al$. For this, we need to additionally restrict values of $\psi$. For simplicity, hereafter we consider $|\psi| < \pi/4$. Then, it is enough to find $\al$'s satisfying
\begin{eqnarray}
    18 \, e^{\al} \Bigg\vert\frac{1 - \al}{\al}\Bigg\vert^{\al} < 1\,.
\label{alphaE}
\end{eqnarray}
If, for instance, we take $\al = \frac{59}{60}$, the left-hand side of \eqref{alphaE} $\approx 0.873 < \frac{7}{8}$ (note that $\frac{59}{60} > \frac{e}{1 + e}$). 
This completes the proof.
\end{proof}

\begin{lemma}
\label{TlessLemma}
For $|\psi| < \pi/4$ and $\al = \frac{59}{60}$, the sum of absolute values of amplitudes of the trees $T_<$ is smaller than the sum of an absolutely convergent series,
\begin{eqnarray}
    \sum_{T_<} \big\vert \cA_{T_<}\big \vert &\leq&
    N^2 \frac{2 (x_4 \cos\psi)^5 \sqrt{\lmbd}}{3 (x_4\sqrt{\lmbd} - 1)}\sum_{v=60}^\infty \frac{1}{v^2}\Bigg(\frac{1}{2}\Bigg)^{v}\,,
\label{TlessBound}
\end{eqnarray}
where
\begin{eqnarray}
    x_4 = \max\{x_3, \frac{2^{30}e^{30}}{\cos\psi} \Big(\frac{3\sqrt{2}}{2\cos\psi}\Big)^{59/2}\} + 1\,.
\label{x4}
\end{eqnarray}
\end{lemma}
\begin{proof}
First of all let us note that $x_4 > x3 \geq x_2 > \frac{1}{|\sqrt{\lmbd}|}$, so \eqref{TlessBound} is well defined. Recall that, for $x \geq x_3$,
\begin{eqnarray}
    \Bigg\vert \frac{\sqrt{\lmbd}}{a} \frac{1}{\cos\psi}\Bigg\vert < \Bigg\vert\frac{(1 - a)}{a \cos\psi}\Bigg\vert < \frac{\sqrt{2}}{\cos\psi}\,.
\end{eqnarray}
Taking this into account, we conclude that the corner operators of trees with less than $\al |V(T)|$ leaves, with $\al = \frac{59}{60}$, are bounded by the bounds for corner operators of trees from $T_<$ with the maximal possible amount of leaves $(\ceil*{\frac{59}{60} |V(T)|} - 1)$.
Note that for any $x > \max\{x_3, \frac{2^{30}e^{30}}{\cos\psi} \Big(\frac{3\sqrt{2}}{2\cos\psi}\Big)^{59/2}\}$, we have
\begin{eqnarray}
    e \Bigg\vert \frac{\sqrt{\lmbd}}{a} \frac{1}{\cos\psi}\Bigg\vert^{\frac{1}{30}} 
    \bigg(\frac{3\sqrt{2}}{2\cos\psi}\bigg)^{\frac{59}{60}} < \frac{1}{2}\,.
\end{eqnarray}
Therefore, fixing $x_4$ as in \eqref{x4}, according to the lemma \ref{cornerLemma}, we obtain
\begin{eqnarray}
&&\bigg|
\Tr\Big[
\mathop{\overrightarrow{\prod}}\limits_{c\in \partial T_{n_l, n_i}\,\text{corner}}
\cC_c(i_c)\Big]
\bigg|
\leq \nonumber\\
&&N  \Bigg\vert \frac{\sqrt{\lmbd}}{a} \frac{1}{\cos\psi}\Bigg\vert^{2 |V(T)| - 2 \ceil*{\frac{59}{60} |V(T)|} - 4} \Bigg\vert\frac{3(1 - a)}{2a \cos\psi}\Bigg\vert^{\ceil*{\frac{59}{60} |V(T)|} - 1} \leq
\nonumber\\
&& N  \Bigg\vert \frac{\sqrt{\lmbd}}{a} \frac{1}{\cos\psi}\Bigg\vert^{\frac{1}{30} |V(T)| - 4} \Bigg\vert\frac{3\sqrt{2}}{2\cos\psi}\Bigg\vert^{\frac{59}{60} |V(T)|} \leq
\nonumber\\
&&\frac{2 (x_4 \cos\psi)^5 \sqrt{\lmbd}}{3 (x_4\sqrt{\lmbd} - 1)}\frac{N}{e^{|V(T)|}} \Big(\frac{1}{2}\Big)^{|V(T)|} 
\, .
\label{eqTreeABound5}
\end{eqnarray}

The number of all LVE trees with $n$ vertices is given by $n^{n-2}$, this amount can be employed as an upper bound for the number of trees $T_<$.
Using the lower bound for the factorial, we obtain
\begin{eqnarray}
    |T_<| \leq \frac{n^{n-2}}{n!} \leq \frac{n^{n-2}}{\Big(\frac{n+1}{e}\Big)^n} \leq \frac{e^n}{n^2}\,.
\label{TlessNumber}
\end{eqnarray}
Combining \eqref{eqTreeABound5} with \eqref{TlessNumber}, and \eqref{eqLVE}, we obtain \eqref{TlessBound}.
\end{proof}
\section{Proof of the Theorem 1}
\begin{proof}
In lemmas \ref{trivialTreeLemma}, \ref{T2Lemma}, \ref{FiniteSumLemma}, \ref{TgeqLemma}, and \ref{TlessLemma} we the proved absolute and uniform in $N$ convergence of the series \eqref{eqLVE}. The terms of \eqref{eqLVE} are analytic functions of $\lambda$, if $\lambda \neq 0$, for any fixed values of $x \geq x_4$ and $|\psi| < \frac{\pi}{4}$ from the representation of the variational parameter $a = x \sqrt{\lambda} e^{i\psi}$. Let us now show that the value $x \geq x_4$ can be chosen independently of $\lambda$. Recalling definitions of $x_1$, $x_2$, $x_3$, we can rewrite $x_4$, as
\begin{eqnarray}
    x_4 = &&\max\Bigg\{\frac{\cos\psi + \sqrt{\cos^2\psi + 8 |\lambda|\cos\psi}}{2 |\sqrt{\lambda}|\cos\psi}, \nonumber\\&&\frac{|\sqrt{\lambda}| + 1}{|\sqrt{\lambda}|}, \frac{2^{30}e^{30}}{\cos\psi} \Big(\frac{3\sqrt{2}}{2\cos\psi}\Big)^{59/2}\Bigg\} + 1\,.
\end{eqnarray}
Obviously, $x_4$ is bounded for $|\lambda| > 0$ and $|\psi| < \pi/4$. Therefore, it is enough to consider the variational parameter $a$ with the fixed value of $x = x_5$, where
\begin{equation}
    x_5 = \max_{|\lambda|, \psi} x_4\,.
\end{equation}

However, we cannot choose the unique value of $\psi$ providing the convergence of \eqref{eqLVE} for all $\phi = \arg \lambda$ in the region $|\phi| < \frac{3\pi}{2}$.
Taking $\psi = 0$, we obtain the absolute and uniform in $N$ convergence of \eqref{eqLVE} and analyticity of terms of \eqref{eqLVE} for $|\phi| < \frac{\pi}{2}$ and $|\lambda| > 0$. Due to the uniqueness of analytic continuation and the overlap with the cardioid domain of the analyticity of the free energy, $F[\lambda, N]$, proven in \cite{Gurau:2014lua}, we conclude that the series \eqref{eqLVE} represents the analytic continuation of the free energy for $|\phi| < \frac{\pi}{2}$ and $|\lambda| > 0$.

Taking $\psi = \pm \psi^*$, $0 < \psi^* < \frac{\pi}{2}$, leads to the convergence of \eqref{eqLVE} and to the analyticity of its terms for $-\frac{\pi}{2} + 2\psi < \phi < \frac{\pi}{2} + 2\psi$, $|\lambda| > 0$, which has at least one line overlap with $|\phi| < \frac{\pi}{2}$, $|\lambda| > 0$. Thus we obtain the analytic continuation of the free energy for $|\phi| < \frac{\pi}{2}$ and $|\lambda| > 0$.
\end{proof}

\end{document}